\documentclass[10pt,letterpaper]{asme2ej}
\usepackage[monochrome]{color} 
\usepackage{amssymb} 
\usepackage{amsmath}
\usepackage{epsfig}
\usepackage{floatrow}

\newtheorem{theorem}{Theorem}

\title{Folding Flat Crease Patterns with Thick Materials}

\author{Jason S. Ku\thanks{Address all correspondence to this author.}
    \affiliation{
    	Field Intelligence Laboratory \\
	Massachusetts Institute of Technology\\
	Cambridge, MA 02139 \\
    Email: jasonku@mit.edu
    }	
}

\author{Erik D. Demaine
    \affiliation{Computer Science and Artificial Intelligence Laboratory\\
	Massachusetts Institute of Technology\\
	Cambridge, MA 02139 \\
	Email: edemaine@mit.edu
    }
}

\begin{document}

\maketitle    

\begin{abstract}
{\it Modeling folding surfaces with nonzero thickness is of practical interest for mechanical engineering. There are many existing approaches that account for material thickness in folding applications. We propose a new systematic and broadly applicable algorithm to transform certain flat-foldable crease patterns into new crease patterns with similar folded structure but with a facet-separated folded state. We provide conditions on input crease patterns for the algorithm to produce a thickened crease pattern avoiding local self intersection, and provide bounds for the maximum thickness that the algorithm can produce for a given input. We demonstrate these results in parameterized numerical simulations and physical models.  }
\end{abstract}

\section*{Introduction}

While much of the research in computational origami applies to folded surfaces with zero thickness (particularly structures that fold flat), modeling folding surfaces with nonzero thickness is of practical interest for mechanical engineering. Design approaches for folding thick material have many varied applications from kinetic architecture \cite{Tachi2} and solar panel deployment \cite{Schenk}, to robotics \cite{Balkcom} and nano-fabrication \cite{In}.  These applications have motivated research into the mathematics and mechanics of rigidly folding thick materials \cite{Huff, Miura1, Tachi1}. We discuss some of the existing techniques for taking into account material thickness in the following section.

In this paper, we propose a new approach for accommodating thickness that modifies certain existing crease patterns into new planar folding patterns, preserving some structure of the old crease pattern while folding a form whose facets are separated from one another in the final state. We describe a systematic and broadly applicable algorithm to transform an input flat-foldable crease pattern into a new crease pattern having a facet-separated, nearly flat folded state.

Our approach for converting flat foldings into facet-separated foldings replaces each flat crease in the input crease pattern by two parallel creases symmetrically offset about the original at a distance proportional to an assigned crease width satisfying certain properties of the original crease pattern. Instead of one crease folding flat with a turn angle of $180^\circ$, the two new creases have a turn angle of $90^\circ$. This crease widening creates difficulties at crease-pattern intersections since the offset creases no longer converge to a point. Material in the vicinity around each crease-pattern vertex is thus discarded to accommodate crease widening. While this modification creates holes in the material, it introduces extra degrees of freedom that can allow the widened creases to fold. Additionally the algorithm identifies and removes some surface material on one side of creases to avoid self-intersections. 

We provide conditions on input flat folded states for the algorithm to produce a thickened crease pattern avoiding local self intersection, namely that crease-pattern faces are convex and creases do not touch the insides of other creases in the input. We also provide bounds for the maximum thickness that the algorithm can produce for a given input. We demonstrate our results in parameterized numerical simulations and physical models. 

\section*{Existing Thick Folding Techniques}

There are many existing approaches that seek to account for material thickness in folding applications, each with their own strengths and weaknesses. We discuss the techniques below, which are also illustrated in order in Figure \ref{fig:hinge}. 

\begin{figure}
\centerline{\psfig{figure=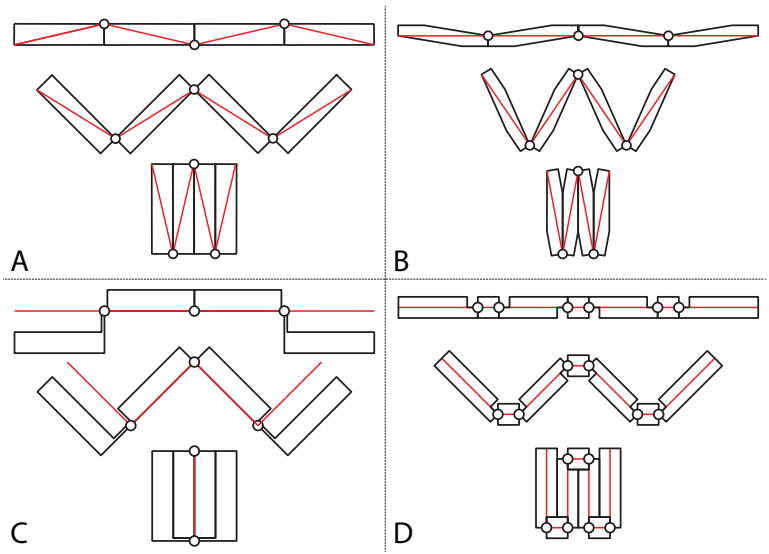,width=4in}\vspace{1pc}}
\caption{Some existing thick folding techniques: (A) Hinge Shift, (B) Volume Trimming, (C) Offset Panel, and (D) Offset Crease. }
\label{fig:hinge}
\end{figure}

\subsection*{A: Hinge Shift}
The hinge shift strategy shifts hinges out of plane to accommodate material thickness~\cite{hoberman2010folding}. While readily useful in creating one-dimensional foldings of thick material, this technique is harder to apply to 2D crease-pattern networks. Hinges start out of plane so cannot build on existing design techniques starting from a coplanar folding pattern. In addition full range of folding motion is restricted. \textcolor{red}{A recent approach extends the idea of hinge shifting to higher degree crease pattern vertices, but this method is geometrically restrictive in the angles and thicknesses allowed~\cite{Chen}.}

\subsection*{B: Volume Trimming}
The strategy presented in~\cite{Tachi2} trims the edges of a thickened surface to overcome many of the difficulties of the hinge shift technique. However, this method also suffers from decreased range of motion and the slanted surfaces can be difficult to fabricate in practice. 

\subsection*{C: Offset Panel}
The offset panel technique~\cite{byu} is probably the most promising in application because it is very flexible while accommodating full range of motion. This method retains hinges at the folding plane but shifts the thick material away from the folding plane. While promising, fabricating such structures can be difficult requiring robust standoffs to connect thick material to hinges. 

\subsection*{D: Offset Crease}
In this paper we expand on the ideas presented in~\cite{zirbel2013accommodating} which accommodates material thickness by widening creases with flexible material, \textcolor{red}{creating a hinge from a two-dimensional region of material.} We propose a modification of the offset crease technique that widens creases in a systematic way, \textcolor{red}{replacing each crease with two ideal hinges} without relying on flexible materials. While this technique does not preserve exact structure of the input crease pattern, it creates a structure that can be easier to fabricate than other techniques. \textcolor{red}{Additionally, the proposed technique allows original facets to be parallel in both flat and folded configurations, potentially allowing for alignment of surface mounted components.} We describe this technique in detail in the following sections, concentrating first on definitions and then the algorithm itself.

\textcolor{red}{Related to the proposed method are a few other methods for accommodating material thickness. A patent by Hoberman~\cite{hoberman1991reversibly} offsets creases in a non-parallel way to accommodate thickness, but also suffers from decrease range of fold angle and does not natural handle crease patterns with internal vertices. Still other methods involve adding degrees of freedom by allowing faces to slide longitudinally along creases, but can be quite difficult to fabricate~\cite{trautz2010deployable}.}

\section*{Definitions}
We would like to take as input a surface that has been folded flat and output a ``thickened'' version. In order to perform this task, we must first specify the input precisely, namely the flat folded state. We will describe input flat folded states by way of crease patterns and \textcolor{red}{non-wrapping} layer ordering graphs. 

Let a \emph{crease pattern} $\Xi$ be a finite straight-line planar graph embedding in $\mathbb{R}^2$. Call crease-pattern edges \emph{boundary edges} if they bound the exterior face, and call them \emph{creases} otherwise. Similarly, call crease-pattern vertices \emph{exterior} if they bound the exterior face with all other vertices \emph{interior}. When we speak of angles around an interior vertex $v$, we are referring to the cyclically ordered set of angles between adjacent edges connected to $v$. A crease pattern is said to be \emph{locally flat-foldable} if the alternating sum of angles around every interior vertex is zero. As discussed later, we will also restrict locally flat-foldable crease patterns to have only convex interior faces.

Certainly if we are given as input a flat folded surface, the network of creases on the unfolded surface define a crease pattern which will be locally flat foldable. The next thing to pin down is the ordering of layers in the folded state.

Given a locally flat-foldable crease pattern $\Xi$, a \emph{flat mapping function} $f_\Xi:\Xi\rightarrow\mathbb{R}^2$ is a piecewise isometric mapping under which each interior face of $\Xi$ is congruent, interior faces that share an edge in $\Xi$ share the same edge in $f_\Xi(\Xi)$, and exactly one of any two adjacent interior faces in $f_\Xi(\Xi)$ is reflected from its orientation in $\Xi$ (i.e. each crease has been folded). This function uniquely exists for a locally flat-foldable crease pattern up to isometry (see Figure~\ref{fig:def}).

\begin{figure}
\centerline{\psfig{figure=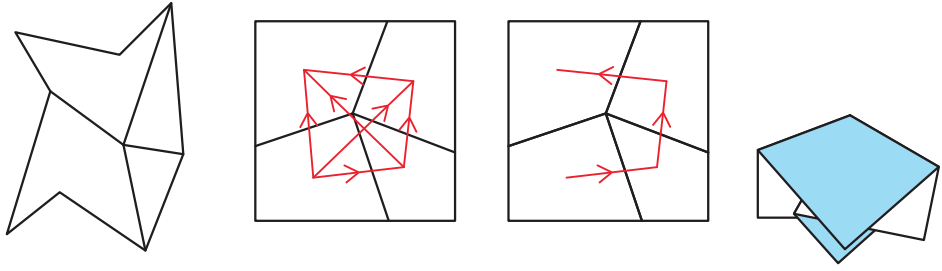,width=5in}}
\caption{From left to right: (1) generic crease pattern $\Xi_0$, (2) locally flat foldable crease pattern $\Xi$ with layer ordering graph $\Lambda$, (3) with reduced layer ordering graph $\Gamma$,  and (4) flat folding $f_\Xi(\Xi)$.}
\label{fig:def}
\end{figure}

Here we adapt the work on layer ordering presented in \cite{TreeMaker_OSME2006}. Given an existing flat folded surface with crease pattern $\Xi$ a \emph{layer ordering graph} $\Lambda$ is a directed graph on the faces of $\Xi$ with an edge between faces $A$ and $B$ if and only if there exists some points $a\in A$ and $b\in B$ such that $f_\Xi(a) = f_\Xi(b)$ (the faces overlap in the folding). The direction of the edges in the directed graph are given by arbitrarily calling the surface normal of some face in the flat folding `up' and drawing edges to point to the face on top of the other. Such a layer ordering may not be well defined if faces are not convex (parts of a face may exist above and below another); as such we will restrict ourselves to crease patterns with convex faces for the remainder of the paper. Additionally, constructing the desired face offset folded state will be impossible if the faces of the layering ordering graph contains a directed cycle because some faces could not be ordered. We will thus restrict to only flat folded surfaces with acyclic layer ordering graphs whose faces can be partially ordered. 

Layer ordering graphs can be very complicated, typically containing edges on the order of the squared number of crease-pattern faces. However, they often contain significant redundancy with respect to providing layer ordering information. For example, consider an edge of a layer ordering graph $(A,B)$ from crease-pattern face $A$ to $B$ ($B$ is on top of $A$), for which there exists some other directed path $L$ from $A$ to $B$. Transitivity ensures that $L$ enforces the ordering condition imposed by $(A,B)$, so edge $(A,B)$ is redundant and can be removed from the graph without losing any layer ordering information. We then implicitly construct the \emph{reduced layer ordering graph} $\Gamma$ from the layer ordering graph $\Lambda$ by identifying any such redundant edge and removing it from the graph. This process terminates and results in a unique output since it is a transitive reduction.


Lastly, we define a flat folded state $(\Xi,\Gamma)$ as a locally flat-foldable crease pattern together with a reduced layer ordering graph free from self-intersection. Specifically, for any crease $\xi$ bounding faces $A$ and $B$ and a third face $C$ which strictly intersects $\xi$, no directed path exists in the reduced layer ordering graph $\Gamma$ between faces $A$ and $B$ visits face $C$ (face $C$ does not intersect crease $\xi$). This object will serve as the input to our thickening algorithm. Note that a flat-folded state implies a crease assignment to each crease (either mountain or valley) by comparing the orientation and order of faces according to the flat mapping function $f_\Xi$ and $\Gamma$. Further, we call the reflex side of a creased surface the \emph{outside} of the crease, and similarly we call the convex side of a creased surface the \emph{inside} of the crease. 

A restriction on our approach is if two creases in a crease pattern wrap around each other in the flat folded state, specifically if one crease touches the inside of another crease, self intersection can become a problem. We will go into more detail as we describe the algorithm, but for now we will call an input flat folded state \textcolor{red}{\emph{non-wrapping}} if no crease \textcolor{red}{or boundary edge point} of the input touches the inside of another crease.

\section*{Algorithm}

The goal of this paper is to construct a thickened version of a given a \textcolor{red}{non-wrapping} flat folded state $(\Xi,\Gamma)$. The strategy is to offset crease-pattern faces from their flat folded state consistent with their layer ordering and create new creases to accommodate the offset. First, we must define an offset distance between every pair of faces which implies a width for each crease. Second, we construct scalable polygons at each interior crease-pattern vertex from which material will be removed to accommodate widened creases. Third, we refine the polygons to ensure that each effective vertex does not exhibit local self intersection. Fourth, we calculate a range for allowable scale factors such that vertex polygons do not intersect. Fifth, we lay out the new crease pattern with holes having a non-flat folded state according to a chosen scale in the allowable range. Last, we address constructing the thickness of each face based on avoiding local self intersection. Additional adjustments may then be made to account for global self intersections.

\subsection*{Step 1: Crease Width}

The first goal of the algorithm is to specify a width for each crease in a flat folded state $(\Xi,\Gamma)$, with all mutually consistent with the layering order of offset faces. Intuitively, we want to separate the layers of the input by nonzero amounts and assign a crease width based on the distance between adjacent faces. If crease widths are chosen small, we can think of the desired output as an ``almost flat" version of the original that allows for nonzero space between layers. The concept of crease width is related to the same term applied to the one-dimensional stamp folding problem \cite{umesato2011complexity}, but we apply it to 2D flat-foldable crease patterns with sortable layer orderings. For our purposes, given reduced layer ordering graph $\Gamma$ it suffices to choose a positive weight for each directed edge such that given any two interior crease-pattern faces $A$ and $B$, every path from $A$ to $B$ in $\Gamma$ has the same weight sum. We will call such a weight assignment $\omega : \xi\in\Xi \rightarrow \mathbb{R}^+$. 

Such a weight assignment always exists; particularly one can be constructed by choosing an arbitrary linearization of the partial order prescribed by $\Gamma$ to create a total order, and defining the weight along a crease to be the absolute difference between the layer ordering numbers of the crease's incident faces. By giving a weight to each crease of $\Gamma$, we can calculate a crease width for every crease of $\Xi$ by summing the total weight along any path from one face incident to the crease, to the other. 

The choice of $\omega$ can be viewed as a design choice for the algorithm implementer. One might strive to choose an $\omega$ that optimizes some natural metric such as minimizing the maximum thickness of any crease, but the work in \cite{umesato2011complexity} and \cite{Thickness_WALCOM2015} seem to suggest such questions may be NP-hard even for one-dimensional graphs. As such, we do not attempt to optimize the choice of $\omega$ here, and leave the exploration in this area as an open problem.


Once we have assigned a crease width to each crease, the construction involves replacing each crease in the input crease pattern with two parallel creases symmetrically offset about the original, separated at a distance proportional to the assigned crease width. This replacement creates difficulties at crease intersections since the offset creases will no longer converge to a point. Material in the vicinity around each crease-pattern vertex will need to be discarded to accommodate the widened creases. Next, we will discuss the construction of the region to be discarded. 

\subsection*{Step 2: Polygon Construction}

Now that crease widths have been defined, we must interface widened creases with each other in the vicinity of crease-pattern vertices. For each vertex, we construct a polygon that will interface with widened crease lines around the vertex. These polygons will be scalable based on how thick we would like to make the material with respect to the crease pattern, up to a point. We will deal with the allowable range of scaling factor later. First, we must define the geometry of these vertex polygons so they will align with all the crease widths around the vertex. 

\begin{figure}
    \centering
    \begin{floatrow}
      \ffigbox[\FBwidth]{
      	\caption{Polygon construction. A generic internal crease pattern vertex showing relationship between offsets and angles.}
	\label{fig:polygon}}{
        \psfig{figure=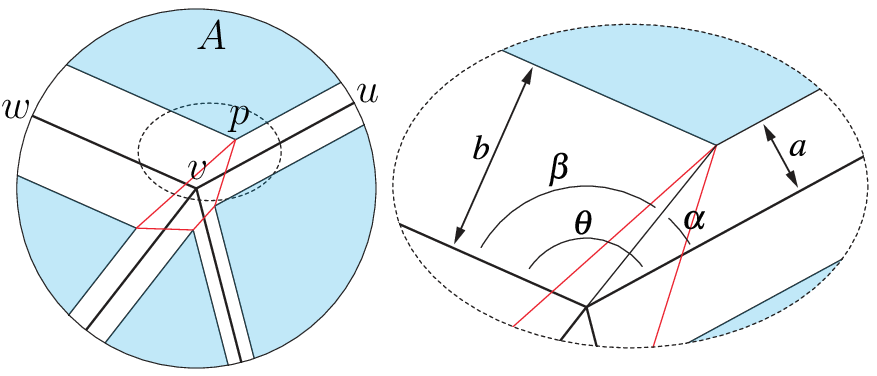,width=0.45\textwidth}
      }
      \ffigbox[\FBwidth]{
      	\caption{A non-simple vertex polygon and refinement by clipping crossings.}
	\label{fig:cross}}{
        \psfig{figure=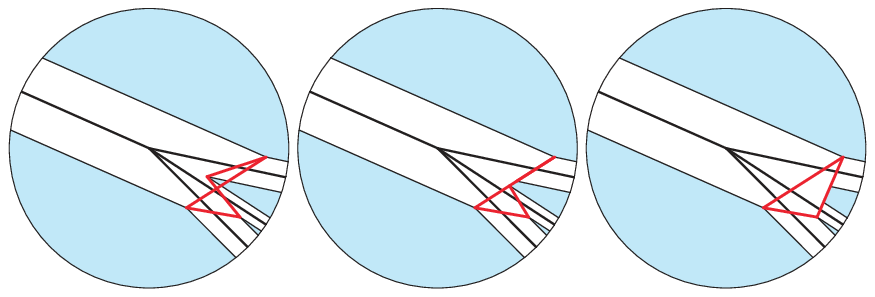,width=0.45\textwidth}
      }
    \end{floatrow}
\end{figure}

We want a vertex polygon to contain one vertex per face adjacent to the crease-pattern vertex at a distance from each adjacent crease proportional to the crease width of the crease. Consider crease-pattern vertex $v$ with face $A$ adjacent to it, bounded by adjacent creases $\{u,v\}$ and $\{v,w\}$ with crease widths $2a$ and $2b$ respectively. Let the angle between these creases be $\theta$. Then the location of the polygon vertex $p$ in this face must be a distance $a$ from crease $\{u,v\}$ and distance $b$ from crease $\{v,w\}$. This point is uniquely defined and can be parameterized by the length $h$ of segment $\{v,p\}$ and the angles $\alpha$ and $\beta$ between this segment and creases $\{u,v\}$ and $\{v,w\}$ respectively (see Figure~\ref{fig:polygon}). Some trigonometry reveals that these angles are given by
\begin{equation}
\tan\alpha = \frac{\sin\theta}{b/a+\cos\theta},\qquad\tan\beta = \frac{\sin\theta}{a/b+\cos\theta}
\end{equation}
with domains $\alpha,\beta\in[0,\pi]$, and $h = a/\sin\alpha = b/\sin\beta$. Repeating this procedure for each face adjacent to an interior crease-pattern vertex constructs points that when connected based on facet adjacency form a polygon.  For exterior crease-pattern vertices, the same construction applies except we include the original vertex and intersections between crease width lines and boundary edges in our polygons. We call the regions in each face bounded by offset creases \emph{reduced faces} (shown in blue in the figures). Unfortunately, edges of a constructed vertex polygon may properly cross as in Figure \ref{fig:cross}. However, we can easily modify the vertex polygon to be weakly simple, or even convex, by clipping any facet sector crossing the polygon. More specifically taking the convex hull of the vertex polygon, mark each vertex whose adjacent reduced face does not properly intersect the convex hull. Trimming the intersecting reduced faces against the convex hull of the marked vertices results in an appropriate convex vertex polygon, though in some cases it may suffice to remove less material (see the middle diagram in Figure \ref{fig:cross}). Note that crossings can only arise if two adjacent vertex angles sum to more than $180^\circ$. The reduced faces of these two angles cannot properly intersect the convex hull of the vertex polygon, so at least two marked vertices exist.

Locally, this polygon divides the area around the vertex into three region types: the polygon, widened creases, and reduced faces (the cardinality of the latter two equaling the number of creases adjacent to the crease-pattern vertex). We will use this terminology to talk about these regions in the following sections.

\subsection*{Step 3: Refinement}

\begin{figure}
    \centering
    \begin{floatrow}
      \ffigbox[\FBwidth]{
      	\caption{Trimming intersecting region.}
	\label{fig:refine}}{
        \psfig{figure=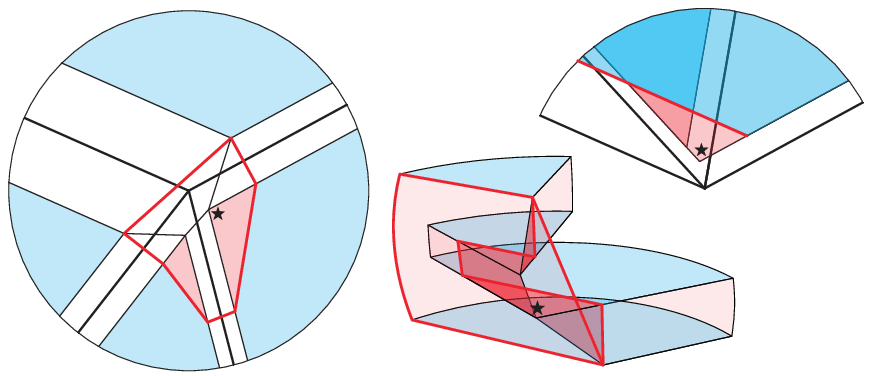,width=0.45\textwidth}
      }
      \ffigbox[\FBwidth]{
      	\caption{Unbounded intersection for inside touching creases in input flat folded state.}
	\label{fig:wrap}}{
        \psfig{figure=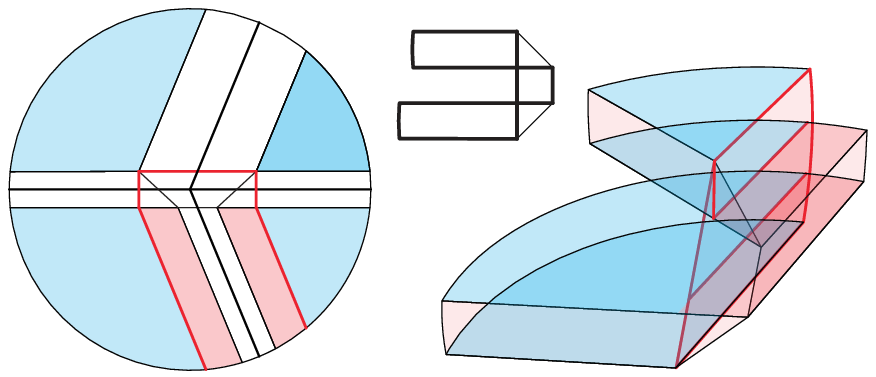,width=0.45\textwidth}
      }
    \end{floatrow}
\end{figure}

The newly constructed creases and polygons in the previous sections serve to locally satisfy isometry between offset faces by removing material at a vertex and adding new creases to accommodate the offset. However, creases with larger crease width require more paper to be absorbed into widened crease regions, reducing the size of surrounding reduced facets. The interaction of this tradeoff between different regions creates the potential for intersection between widened creases and reduced facets. We fix this type of self intersection by checking each widened crease/reduced facet pair for intersection. If they intersect, trim the reduced facet along the widened crease boundary and refine the vertex polygon to reflect this change (see Figure~\ref{fig:refine}). 

There is a worry that this procedure could remove material that is not a bounded distance from the vertex. For example, the crease pattern shown in Figure \ref{fig:wrap} contains two creases that when widened have an intersection that extends to infinity. Fortunately, this type of situation only occurs locally when some crease of the input touches the inside of another crease, which we have forbidden by requiring a \textcolor{red}{non-wrapping} input. Reduced facets can only be trimmed a finite number of times because trimming cannot increase the number of intersections, thus the refinement terminates.

\subsection*{Step 4: Scale Factor}

\begin{figure}
\centerline{\psfig{figure=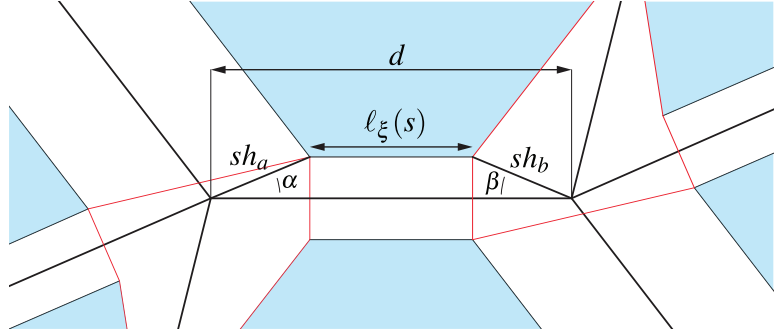,width=4in}}
\caption{Scale factor calculation showing relevant quantities.}
\label{fig:scale}
\end{figure}

After creating vertex polygons and local widened crease/reduced facet regions that locally do not self intersect, we can determine how large these polygons can be before intersecting each other. Each widened crease edge is bounded on either side by a vertex polygon. Consider crease $\xi$ with length is $d$. Then each widened crease edge of $\xi$ is shorter than $d$ according to the size of each incident vertex polygon. Let $(h_a,\alpha)$ and $(h_b,\beta)$ define the locations of the vertex polygon vertices on either side of $\xi$ contained in the same face $F$. If we let the size of all vertex polygons scale by a factor $s$, then the length $\ell_\xi$ of the widened crease segment in $F$ is given by the following function of $s$ (see Figure~\ref{fig:scale}):
\begin{equation}
\ell_\xi(s) = d-s(h_a\cos\alpha + h_b\cos\beta).
\end{equation}
For $(h_a\cos\alpha + h_b\cos\beta)$ negative, $\ell_\xi(s) > 0$ for all $s>0$ so this crease $\xi$ does not restrict scale. For $(h_a\cos\alpha + h_b\cos\beta)$ positive, there exists some $s_\xi$ strictly positive for which $\ell_\xi(s_\xi) = 0$. This event corresponds to neighboring vertex polygons intersecting which we would like to forbid. Taking the minimum $s_\xi$ over all creases $\xi\in \Xi$ yields a strictly positive upper bound $s^*$ on scale factors by which vertex polygons can be scaled without overlap. Note that for $s=0$, the crease pattern is not offset at all and facets remain coplanar, and the folded form cannot be produced with material of any finite thickness. Strictly positive $s$, such as $s^*$ calculated above, allow the modified pattern to accommodate some finite thickness, with a larger $s$ accommodating a larger thickness relative to the geometry of the input crease pattern. Of course this scale $s^*$ only insures that vertex polygons do not interact. It is possible that 
with this calculated scale, global intersection between faces of the folding can still arise. Nonetheless, we show the following:

\begin{theorem}
Given a \textcolor{red}{non-wrapping} flat folded state $(\Xi,\Gamma)$ and weight assignment $\omega : \xi \in \Xi\rightarrow\mathbb{R}^+$, 
there exists some positive non-zero scale $s$ for which the above construction globally contains no strict intersection
between faces in the three-dimensional folded state.
\end{theorem}

\begin{proof}
Suppose for contradiction that the construction produces intersecting faces for every positive non-zero scale $s$.
Intersection cannot occur between reduced polygon faces because
they are offset from each other in a way consistent with the input \textcolor{red}{non-wrapping} reduced layer ordering graph
containing no self-intersection. Thus, any face-face intersection must exist between a widened crease face and some other
face. Let $\xi$ be the original crease corresponding to some widened crease face strictly interesting another face $F$.
Because the input is \textcolor{red}{non-wrapping} containing no self-intersection, $\xi$ does not intersect $F$ in the input \textcolor{red}{non-wrapping} flat folded state.
Increasing the scale $s$ from $0$ and performing the above construction results in a continuous parameterized 
family of three-dimensional foldings. More importantly, let $d(s)$ be the minimum distance between the widened 
crease associated with crease $\xi$ in a construction with scale $s$ and the reduced polygon formed from face $F$. 
Then $d(s)$ is positive for $s=0$ and varies continuously and weakly monotonically with $s$.
Thus there exists some positive non-zero scale $s'\in(0,s^*)$ for which $d(s')$ is also positive. But this is true for
every intersection involving a crease, so there must exist some scale where no intersection occur, a contradiction.
\hfill $\square$
\end{proof}

\subsection*{Step 5: Final Construction}
Now given a flat folded state $(\Xi,\Gamma)$ and width assignment $\omega$, we can calculate the upper bound $s^*$ on scale to forbid vertex polygon intersection and choose a scale $s'$ in the range $(0,s^*)$ to construct a modified crease pattern that avoids self intersection. Quite simply the construction is placing vertex polygons scaled by $s'$ and adding widened crease lines parallel to the original creases between vertex polygons. The entire process is shown in Figure \ref{fig:construct}: \textcolor{red}{first the input non-wrapping flat folded state, offset facets, and finally the offset polygons, together with their counterparts in the folding domain. }

\begin{figure}
    \centering
    \begin{floatrow}
      \ffigbox[\FBwidth]{
      	\caption{Construction process.}
	\label{fig:construct}}{
	\psfig{figure=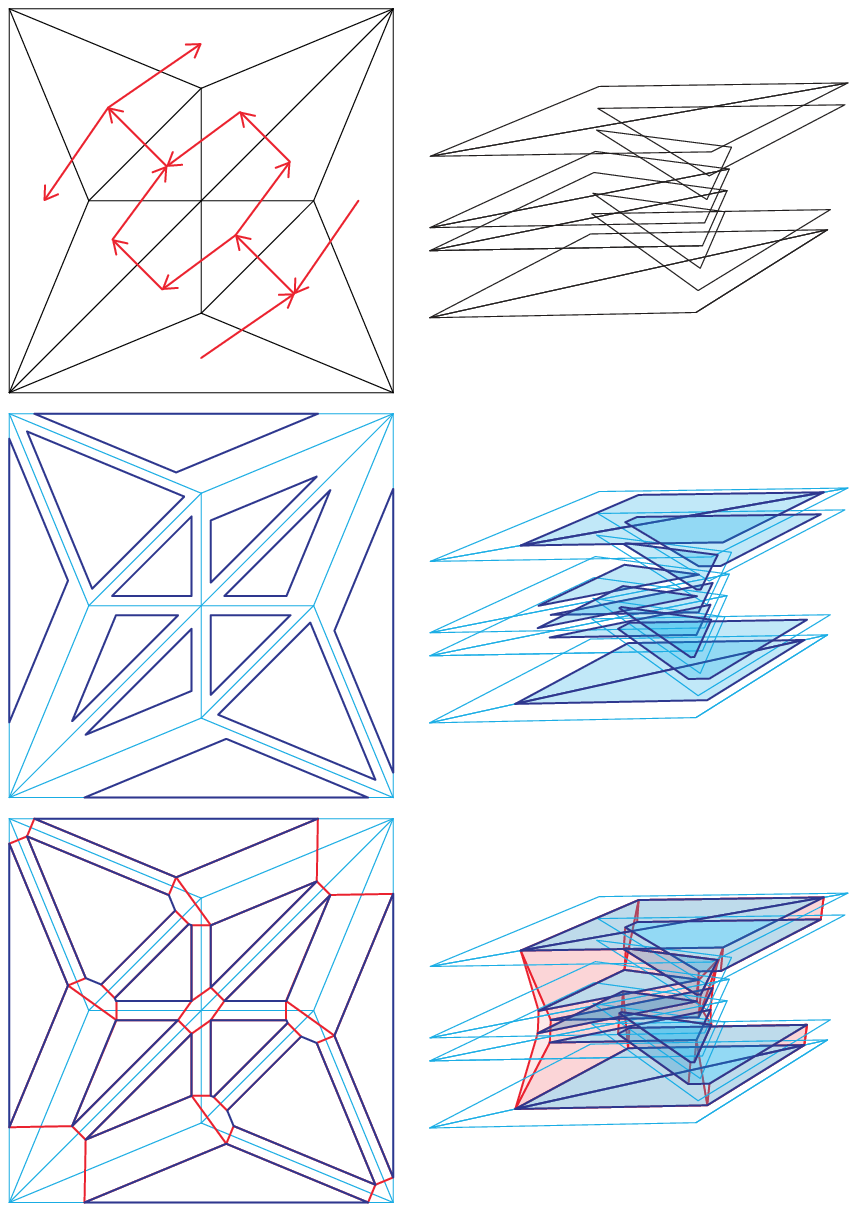,width=0.45\textwidth}
      }
      \ffigbox[\FBwidth]{
      	\caption{Foam core model of a thickened traditional bird base.}
	\label{fig:model}}{
        \psfig{figure=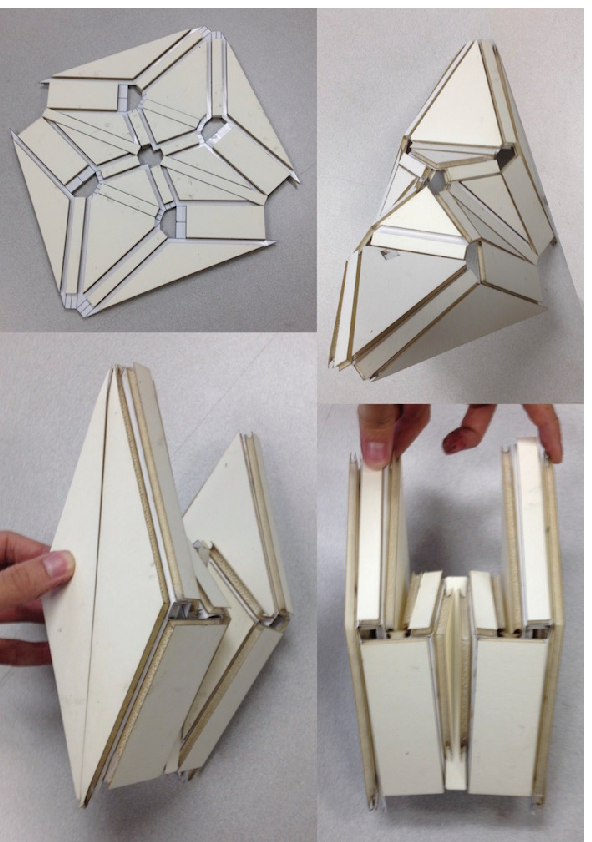,width=0.45\textwidth}
      }
    \end{floatrow}
\end{figure}

\begin{theorem}
Given a \textcolor{red}{non-wrapping} flat folded state $(\Xi,\Gamma)$ and weight assignment $\omega : \xi \in \Xi\rightarrow\mathbb{R}^+$, 
the construction above terminates in polynomial time.
\end{theorem}

\begin{proof}
Given the weight assignment, the vertex polygons are each bounded, 
constructed as described in Step 2 by offsetting the
original geometry by finite amounts and connecting vertices,
which can be constructed directly in linear time. 
Clipping ensures the vertex polygons are weakly simple and can
be performed naively by comparing each vertex-edge pair in quadratic time.
Trimming in Step 3 can also be implemented in quadratic time by checking
each pair of faces locally around a vertex. Local intersections of the faces around
a vertex are thus removed in the trimming step by construction. Calculating the
scale upper bound $s^*$ guaranteeing that vertex polygons do not intersect 
requires a constant-sized evaluation per edge, while
calculating an appropriate $s'$ can be computed by evaluating the appropriate
scale for each possible intersection pair in at most quadratic time, 
and choosing the minimum scale. Thus the procedure
can be implemented to terminate in quadratic time which is polynomial. \hfill $\square$
\end{proof}

\subsection*{Adding Thickness}
The above construction creates a modified thin crease pattern that separates faces in the folded form to make room for thick panels. Adding material to the constructed thin surface is relatively easy. In general, if crease widths are chosen arbitrarily, facets can be assigned a range of thicknesses to either side that can be accommodated by the crease widths. However, a simpler and more practical assignment might be to assign the same max thickness to the entire crease pattern as many manufacturing processes could benefit from this kind of uniformity (nano-fabrication, sheet metal construction, etc.). We can simply define the max panel thickness $t_{max}$ as the smallest crease width assigned to the flat folded state. 

However, this panel thickness cannot be added everywhere or material would self-intersect. 
For example, if finite panel thickness exists everywhere on adjacent faces on the inside valley side of each crease, 
the crease would not be able to fold at a right angle without the added material intersecting when folded. 
There are many ways to solve this problem by removing material. 
We suggest keeping full panel thickness on widened crease regions to strengthen these traditionally weak interfaces,
and removing material from the adjacent face incident to the crease. To accommodate widened crease panel thickness on both sides, we must remove a strip of material of width $t_{max}/2$ on either side of the widened crease from the reduced facets adjacent to the crease, only on the crease's inside surface. This modification will ensure that material in the vicinity of creases do not locally self-intersect. 

The problem of global material self intersection during a folding motion is a more difficult computational task, though there are existing computational methods for addressing this issue. The offset panel techniques of \cite{byu} also point out this problem. We are looking into more efficient techniques to perform global folding motion collision detection to aid real-world design applications. 

\section*{Models}

\begin{figure}
    \centering
    \begin{floatrow}
      \ffigbox[\FBwidth]{
      	\caption{Numerical folding simulation of two thickened crease patterns using Freeform Origami.}
	\label{fig:freeform}}{
	\psfig{figure=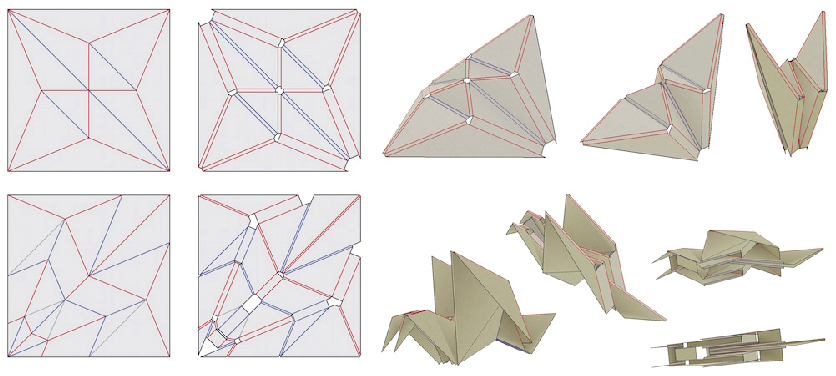,width=0.5\textwidth}
      }
      \ffigbox[\FBwidth]{
      	\caption{Parameterized thick single vertex construction in Mathematica.}
	\label{fig:mathematica}}{
        \psfig{figure=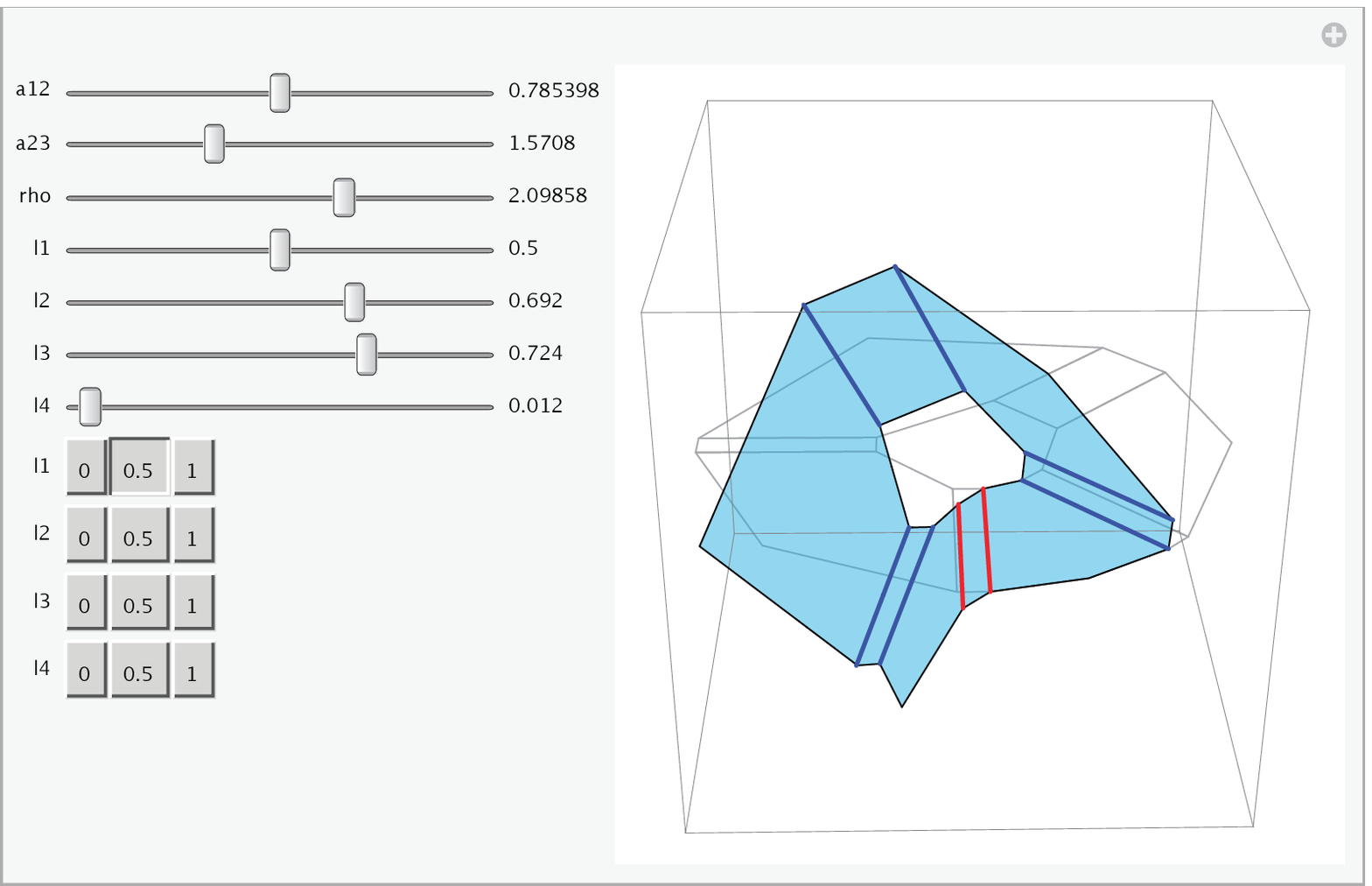,width=0.4\textwidth}
      }
    \end{floatrow}
\end{figure}

We developed numerical and physical models to demonstrate the algorithm presented above. We used the algorithm described to modify two existing rigid-foldable flat-foldable crease patterns, the traditional bird base and a modified rigid foldable flapping bird designed by Robert Lang as shown in Figure \ref{fig:freeform} \textcolor{red}{(bird base on top, and flapping bird on bottom). On the left the original and modified crease patterns are shown, followed by snapshots of each crease pattern folding.} These modified crease patterns were input into a numerical origami simulator called Freeform Origami \cite{Tachi_freeform}. This simulator is able to fold a crease pattern incrementally through rigid folding configuration space while seeking to maintain developability and planarity constraints converging iteratively to within double precision.  Folding these crease patterns in the simulator demonstrated multiple rigid folded states throughout the folding process to very high accuracy. These simulations provide evidence that a path through the configuration space exist for complex crease patterns between the unfolded and folded states produced by this algorithm. \textcolor{red}{Such a movement seems possible for single-vertex crease patterns because the number of degrees of freedom of the modified structure should in general increase.}

We also used a Mathematica model shown in Figure \ref{fig:mathematica} to apply the algorithm to single-vertex crease patterns to try and find a path in the folding configuration space between the unfolded and folded states produced by this algorithm. \textcolor{red}{The model allows the user to change the parameters of the system, namely fold angles between creases and splitting ratios between offset crease pairs, in order to satisfy closure. While we have not found an analytic closure constraint relating fold angles and splitting ratios, we have been able to achieve closure numerically to double precision for a range of inputs.} Our results in this area are preliminary, but we have experimental evidence to support that single-vertex crease patterns thickened with this algorithm have a rigid foldable path between unfolded and folded states. We conjecture that the state space for thickened single-vertex crease patterns is a sphere embedded in the multidimensional parameterized space and will leave further discussion in this area to future work. 

Lastly, a physical model of a thickened version of the traditional bird base was fabricated using 3/8'' foam core pasted on either side of thin paper. Some views of the physical model can be seen in Figure \ref{fig:model}. The folding action observed with this model agrees well with the folded states of numerical simulation, and the motion feels tightly constrained in contrast to the folding mechanisms described in \cite{zirbel2013accommodating}. Empirically fixing the dihedral angle between sector faces while adjusting the angle ratio at one crease, a continuous adjustment of the other crease ratios was observed, also supporting the spherical configuration space conjecture.

\section*{Conclusion}

In this paper we have presented a new method for creating thick folded structures from flat folded states. The algorithm proposed has many benefits over existing thick folding techniques. Facet surfaces in the produced structure's unfolded state are coplanar allowing for ease of fabrication in layer-by-layer manufacturing processes. These same surfaces are parallel in the produced structure's folded state allowing any surface mounted components to mate naturally. Panel thicknesses can be adjusted according to material and scale within bounds provided by the algorithm. Further, every finite area of the algorithm's produced surface may be assigned non-zero thickness, allowing for the production of strong and tightly constrained mechanisms.  

The offset crease method provides a thickened folded state suggesting a full range of folding motion as well as parallel facets when fully folded. Assigning crease widths to comply with the acyclic layer ordering of the input flat folded state provides a flexible design space for varied applications, while still constructing one non-trivial folded state with planar facets. While it is still open whether a path of rigid folded states exists through the configuration space in general, there is evidence that one exists for single vertex crease patterns given our numerical models. While compensating for material thickness is not as difficult for non-flat foldings, because faces do not meet each other when folded, we are exploring ways of extending this method for non-flat foldings, particularly those containing face-to-face contact in their non-flat folded form.

\bibliographystyle{asmems4}

\begin{acknowledgment}
Supported in part by NSF ODISSEI grant EFRI-1240383 and NSF Expedition grant CCF-1138967. 
\textcolor{red}{We also thank our anonymous reviewers for helpful comments and additional references.}
\end{acknowledgment}

\end{document}